\documentclass{article}  

\usepackage{graphicx}
\usepackage{latexsym}
\usepackage{fullpage}

 \usepackage[utf8]{inputenc} 
 \usepackage[T1]{fontenc}    
\usepackage{hyperref}       
 \usepackage{url}            
\usepackage{amsfonts}       
 \usepackage{nicefrac}       
 \usepackage{microtype}      
 \usepackage{xcolor}         

 \usepackage{times}
 \usepackage{soul}
 \usepackage[small]{caption}
 \usepackage{graphicx}
 \usepackage{amsmath}
 \usepackage{amsthm}
 \usepackage{booktabs}
 \usepackage{algorithm}

\usepackage{authblk}

\usepackage[american]{babel}
\usepackage{mathtools} 
\usepackage{tikz} 

\usepackage[textwidth=2cm]{todonotes}

\usepackage{tikz}

\usepackage{commath}
\usepackage{algpseudocode}

\newtheorem{theorem}{Theorem}

\newtheorem{lemma}{Lemma}
\newtheorem{proposition}[lemma]{Proposition}
\theoremstyle{definition}

\newtheorem{definition}{Definition}

\newcommand{\sfNP}{\textsf{NP}}

\newcommand{\PFM}{{\em Proportionally Fair Matching}\xspace}
\newcommand{\DM}{{\em Diverse Matching}\xspace}

\usepackage{xspace}
\newcommand{\matchlb}{\DM}
\newcommand{\propfair}{\PFM}

\newcommand{\sol}{\mathcal{S}}

\title{Online Algorithms for Matchings with Proportional Fairness Constraints and Diversity Constraints\footnote{Author names are listed in alphabetical order.}}

\author[1]{Anand Louis}
\author[2]{Meghana Nasre}
\author[3]{Prajakta Nimbhorkar}
\author[4]{Govind~S.~Sankar}

\affil[1]{%
    Indian Institute of Science, Bangalore\\
    Bengaluru, India
    \thanks{anandl@iisc.ac.in},
}

\affil[2]{%
    Indian Institute of Technology Madras\\
    Chennai, India
    \thanks{meghana@cse.iitm.ac.in}
  }
 \affil[3]{%
    Chennai Mathematical Institute\\
    Chennai, India
    \thanks{prajakta@cmi.ac.in}
}
\affil[4]{%
    Duke University\\
    Durham, USA\\
    \thanks{govind.sankar@duke.edu}
}
\begin{document}
\maketitle
\begin{abstract}
Matching problems with group-fairness constraints and diversity constraints have numerous applications such as in allocation problems, committee selection, school choice, etc. Moreover, online matching problems have lots of applications in ad allocations and other e-commerce problems like product recommendation in digital marketing.
 
We study two problems involving assigning {\em items} to {\em platforms}, where items belong to various {\em groups} depending on their attributes; the set of items are available offline and the platforms arrive online. In the first problem, we study online matchings with {\em proportional fairness constraints}. Here, each platform on arrival should either be assigned a set of items in which the fraction of items from each group is within specified bounds or be assigned no items; the goal is to assign items to platforms in order to maximize the number of items assigned to platforms. 

In the second problem, we study online matchings with {\em diversity constraints}, i.e. for each platform, absolute lower bounds are specified for each group. Each platform on arrival should either be assigned a set of items that satisfy these bounds or be assigned no items; the goal is to maximize the set of platforms that get matched. We study approximation algorithms and hardness results for these problems. The technical core of our proofs is a new connection between these problems and the problem of matchings in hypergraphs.
 
Our experimental evaluation shows the performance of our algorithms on real-world and synthetic datasets exceeds our theoretical guarantees.

\end{abstract}


\section{Introduction}\label{sec:intro}
Matchings in graphs is an important problem in both theory and practice, and has received a
lot of attention in literature over several years. Computing a maximum matching in a bipartite graph under various constraints is the core of many
allocation applications like scheduling~\cite{venkat_scheduling}, school choice~\cite{abdulkadiroglu2003}, 
ad-auctions \cite{mehta_online_survey,mehta_online_adwords}, resource allocation \cite{halabian_resourceallocation},  
healthcare rationing~\cite{AB21} etc. The terminology in these papers varies, and we refer to the two parts of the
underlying bipartite graph with the terms {\em items} and {\em platforms}. In real-world applications, items may
be classified into different groups based on various properties they possess. Seeking to optimize just the cost or
utility of an allocation can be unfair to certain groups or may not serve the intended purpose. For instance, while forming committees from a pool of candidates, it is necessary that each committee contains candidates with expertise from all the relevant areas. This is the case, for example, while forming program committees of conferences, teams to work on projects, or expert committees to evaluate project proposals where it is required that a minimum number of members are picked from each of the relevant sub-areas. Additionally, each committee may have a limit on the maximum number of members it can accommodate from a particular sub-area.


In this paper, we consider the scenario where {\em items} need to be assigned to {\em platforms}, items are classified into groups, and the platforms have certain {\em fairness constraints} on the set of items that gets assigned to them from each group. 
The fairness constraints can be specified in terms of lower bounds on the number items that get assigned to a platform from each group.
However, in many applications, stating the fairness constraints in terms of absolute lower bounds is inadequate. For example, in case of school choice, the total number of applications may not be the same each year (see e.g. \cite{case-study}). Thus, some schools may not fill up all their seats, and specifying constraints on the number of students with each race, ethnicity, and economic background in terms of absolute values do not achieve the desired balance. Similarly, while selecting committees, the number of available candidates and their backgrounds may vary each year. To address this, we consider the notion of {\em proportional fairness}. Formally, our model and the problem definition are as follows:

\noindent {\bf Our model: }
 The input instance consists of a bipartite graph $G=(A\cup P, E)$ where $A$ is the set of items and $P$ is the set of platforms, and $(a,p)\in E$ if and only if item $a$ can be assigned to platform $p$. Let $N(p_j)$ denote the set of items adjacent to $p_j$ in $G$. Moreover, depending on the properties or attributes they possess, the items are classified into $m$ groups $C_1,\ldots,C_m$ where each $C_i\subseteq A$. 

 An assignment or matching $M$ of items to platforms is a subset of $E$. For a platform $p_j$, define $M_j=\{a\in A \mid (a,p_j)\in M\}$. Thus $M_j$ denotes the set of items assigned to $p_j$ in $M$.

We define two problems on this model, depending on the way fairness constraints are specified. The following problem has {\em proportional fairness constraints}:

 \noindent {\bf Proportionally fair matching problem: }
 The input instance is as described above.
 Every platform $p_j\in P$ has a lower bound $\ell_j$ and upper bound $u_j$, respectively denoting the minimum and maximum number of items that can be assigned to $p_j$. 
 Further, each platform $p_j$ has associated balance	parameters $\alpha_j^{(i)},\beta_j^{(i)}$ for each group $C_i$. It is essential that, when $N(p_j)\cap C_i=\emptyset$, $\alpha_j^{(i)},\beta_j^{(i)}=0$.
A platform $p_j$ is said to be {\em satisfied} by $M$ if the following holds:
 \begin{eqnarray*}
 \ell_j  \leq  \abs{M_j} &\leq & u_j \textrm{ and }
 \textrm{ for every }i\textrm{ s.t. }C_i\cap N(p_j)\neq \emptyset\\
\alpha_j^{(i)} \abs{M_j} &\leq & \abs{M_j\cap C_i}  \leq \beta_j^{(i)} \abs{M_j}
 \end{eqnarray*}
The goal is to compute an assignment $M$ of items to platforms such that the number of items assigned to {\em satisfied} platforms is maximized.
 
Additionally, we also consider the problem where the fairness constraints are specified in terms of absolute lower bounds, and the goal is to maximize the number of platforms whose lower bounds are met. We refer to this as the {\em diverse matching} problem. The objective of maximizing the number of satisfied platforms is motivated by its real-world applications like committee selection, where a committee cannot be set up unless all the constraints are met, or setting up teams to work on projects, as experts from all the areas are necessary for the completion of a project. The objective of maximizing the number of committees or teams formed is a natural one here.
The problem is formally defined below. 

\noindent{\bf Diverse matching problem: }
The input instance is the same as described previously.
Each platform specifies a lower bound $\ell_j^{(i)}$
on the number of items it needs from the groups $C_i$. 
It is essential that $\ell_j^{(i)}\leq |N(p_j)\cap C_i|$.
A platform $p_j$ is said to be satisfied if,
for the set of items $M_j$ assigned to it, the following holds: 
\begin{align*}
		\abs{M_j}\geq  \ell_j, \qquad \mathrm{ and } \qquad 
\forall~i,\abs{M_j\cap C_i}\geq \ell_j^{(i)}.
\end{align*}
The goal is to compute an assignment of items to platforms
such that the number of {\em satisfied } platforms is maximized. 

\noindent{\bf The online setting:}
 In many practical applications like ad-allocation, selecting teams from available pool of candidates, assigning item reviews to customers, ride sharing etc., all the platforms may not be known in advance. Our algorithms have the added advantage that they work in this setting where platforms arrive online over time and items are known in advance.
 
 Even though our problem formulations are of immense importance as exemplified by the numerous aforementioned special cases they generalize, there doesn't seem to be any prior work studying these problems at the level of generality of our formulations; see Section \ref{sec:related} for brief summary of the related work.
\subsection{Our Contributions}
We give hardness and approximation algorithms for the \propfair\ and \matchlb\ problems. 
Recall that $\ell_j$ is the minimum number of items required to be assigned to a platform $p_j$, and define $\ell=max_j \ell_j$
taken over all platforms. We state our results for the online setting, clearly they also apply to the {\em offline} setting, where the platforms and items both are known upfront. 

The main result of our paper is as follows:
\begin{theorem}\label{thm:prop-fair}
There is an $O(n^2)$-time online algorithm that outputs an assignment $M$ of items to platforms such that the number of items that get assigned in $M$ is a $2(\ell+1)$-approximation to the optimum solution of the \propfair\ problem, and the fairness constraint for any platform $p_j$ may be violated by at most a fraction $O\big(\frac{1}{\ell_j}\big)$.

Formally, for a group $C_i$ and a platform $p_j$ we have
%
\begin{align*}
		 \abs{M_j}\left(\alpha_j^{(i)}-\frac{3}{\ell_j}\right)   \leq & \abs{M_j\cap C_i} \leq   \abs{M_j}\left(\beta_j^{(i)}+\frac{3}{\ell_j}\right)
\end{align*}
 Here $n$ is the number of vertices in the underlying bipartite graph $G$. 
\end{theorem}
The theorem is proved in Section~\ref{sec:prop-fair}.

We give a slightly better approximation guarantee for the \matchlb\ problem, without any violation of constraints, even when an item can belong to multiple groups. Following is our result for the \matchlb\ problem, proved in Section~\ref{sec:lb}.
\begin{theorem}\label{thm:intro-matching-lb}
There is an online algorithm with competitive ratio $(\ell+1)$ for the \matchlb\ problem 
where $\ell=\max_{p_j}( \max( \ell_j, \sum_i \ell_j^{(i)} ))$.
\end{theorem}

We show that the approximation ratios in Theorems \ref{thm:prop-fair} and \ref{thm:intro-matching-lb} are almost tight by showing that it is NP-Hard to get an approximation ratio larger than $O(\frac{\ell}{\ln \ell})$. Theorem~\ref{thm:hardness-kdim} gives hardness of approximation for the \propfair\ and \matchlb\ problems. The proof is given in Section~\ref{sec:hardness}.  

\begin{theorem}\label{thm:hardness-kdim}
	The \propfair\ and the \matchlb\ problems are {\sf NP}-hard to approximate in polynomial time to within a factor of $O\left (\frac{\ell}{\ln \ell}\right )$ where $\ell=\max_{p_j} \ell_j$. The hardness result holds even when there is a trivial group containing all the items.
\end{theorem}

\noindent {\bf Experimental evaluation:} An experimental evaluation of our algorithms shows that their performance on  real-world and synthetic datasets significantly exceeds our theoretical guarantees. 
Our synthetic datasets are generated using a model that 
loosely resembles a random graph generated from an Erd\H{o}s-R\'enyi
model, and show that our algorithms  
perform very well even for small values of average degree of items.
This performance is explained by the high theoretical guarantees obtained in Theorem~\ref{thm:random} above. 
In particular, our algorithm for the \propfair\ problem outperforms the optimal solution in terms of size, owing to the allowed violation in fairness constraints.

Motivated by the above experimental results, we analyze our algorithm for the \matchlb\ problem on {\em Erd\H{o}s-R\'enyi random graphs}.
This is similar in spirit to the work of \cite{DyerFP93}, where they show that the greedy matching algorithm obtains an almost-optimal matching on Erd\H{o}s-R\'enyi random graphs, which are used as one of the models for real-world instances.
In our theoretical analysis of the algorithm for the \matchlb\ problem on an Erd\H{o}s-R\'enyi random graph,
we see a similar behaviour in the presence of lower bounds.
The random graph model involves a bipartite graph $G=(A\cup P, E)$ where for every $a\in A, p\in P$,
the edge $(a,p)$ exists in $E$ with probability $\rho$. 
We consider a simplistic scenario where there are $\Delta$ groups, each item belongs to one of the $\Delta$ groups, and the lower bound of each group is $\ell$.

\begin{theorem}\label{thm:random}
For any constant $\epsilon>0$, the greedy algorithm from Theorem~\ref{thm:intro-matching-lb} achieves
a $(1-\epsilon)$-approximation with probability $\big(1-\frac{1}{n^c}\big)$
for instances of the \matchlb\ problem
where 
all the lower bounds are fixed to a constant $\ell$, and 
the underlying graph is an
Erd\H{o}s-R\'enyi random graph
with an edge probability of $\rho=\Omega(\frac{\log n}{n})$. 
\end{theorem}
Here, the constant $c$ depends on $\rho$ and $\epsilon$. The proof appears in Section~\ref{sec:exp-random-graphs}.
\subsection{Our Techniques}
The technical core of our theoretical results is a new connection between our problems and the {\em hypergraph matching problem} defined below.
\begin{definition}[Hypergraph Matching]
Given a $k$-uniform hypergraph 
$H=(V, E)$, find the largest matching, viz. a subset of edges that do not intersect. 
\end{definition}
In this paper, we use the folklore greedy algorithm
for hypergraph matching which involves repeatedly choosing a hyperedge
that is disjoint from the edges already included in the
matching, and adding it to the matching.  
\begin{proposition}\label{prop:folklore-greedy}
The greedy algorithm achieves an approximation
ratio of $\Delta$ on hypergraphs with hyperedges
of size $\leq \Delta$.
\end{proposition}

\subsection{Related Work}\label{sec:related}
\textcolor{black}{The  problem of matchings with fairness constraints has been well-studied in recent years} and the importance of fairness constraints
has been highlighted in literature e.g. \cite{halevi_fair_allocation,luss_leximin_fair,devanur_ranking,vishnoi_fairness,KMM15,CHRG16,BCZSK16}. 
A lot of work in literature
has focused on group fairness constraints, modeled as upper bounds on the number of items from each group
that can be allocated to a platform. This is referred to as {\em restricted dominance} in literature \cite{bera_fair_clustering}. 
Constraints that are specified in terms of lower
bounds constraints have also been considered for different problems like clustering, minimum cut, knapsack (see e.g. \cite{AS21,BDSSTV21,LCCZ20} for some recent results). Bera et. al
\cite{bera_fair_clustering} consider {\em proportional fairness constraints} for clustering problems, and give an algorithm with additive violation of constraints. In this work, we consider them in the context of matchings. In ~\cite{SLNN21} the authors consider upper bounds for problems similar to ours. However, the techniques required for
the setting with lower bounds in our work are different from the ones used in \cite{SLNN21}.

There has been more recent work where different models of fairness in matchings have been considered. In \cite{GSB20}, individual fairness is addressed, in \cite{BDLE20}, two-sided fairness is considered in terms of utilities, whereas in \cite{MXX21}, group-fairness in terms of the minimum service rate across all groups is studied. In \cite{esmaeili2022rawlsian}, the authors consider group and individual Rawlsian fairness criteria in an online setting, where fairness is attained at the cost of a drop in operator's profit. In \cite{EfthymiouSPC21}, group fairness constraints have been considered for the {\em entity-resolution} problem.

Proportional fairness  has been considered for the {\em candidate selection problem} in \cite{BeiLPW22}.
In their setting, the input consists of $n$ candidates and $m$ properties. Each candidate can posses a subset of the $m$ properties and a property $p_j$ is associated with proportional fairness constraints $\alpha_j$, $\beta_j$. Additionally, the input contains a threshold $k$. The goal is to output a set of at most $k$ candidates which satisfies the proportionality constraints. The authors
show that even deciding whether there exists a non-empty feasible set of candidates is NP-hard. They complement the hardness by providing polynomial-time approximation algorithms with a slight violation of the proportional fairness constraints. 

Our proportionally fair matching problem can be viewed as selecting multiple committees instead of a single committee / subset. Thus their hardness and inapproximability hold for our problem as well. We therefore consider a special case where every candidate belongs to exactly one group. In this case, if we are interested in selecting a single committee, the problem becomes tractable, however, it remains NP-hard when there are multiple committees to be selected (Theorem~\ref{thm:hardness-kdim}).
An advantage of our algorithms is that they simply follow the greedy paradigm, and work in online setting as well. On the contrary, those of Bei et al. \cite{BeiLPW22} involve solving an LP or ILP.

In a recent work, \cite{BFIS23} consider the proportional fairness model in case of non-bipartite graphs. Their model involves colors on edges, and the proportional fairness constraints are stated in terms of two parameters $\alpha,\beta$. The goal is to construct a matching that has at least $\alpha$ and at most $\beta$ fraction of edges of each color. The running time of their algorithm is exponential in the number of colors, and the violation of proportional fairness parameters also depends on the number of colors. In another work, \cite{PLN22}, the authors consider fairness constraints in terms of lower and upper bounds for each group, and additionally consider {\em individual fairness constraints} denoting the bounds on probabilities with which an item must be matched to a subset of platforms. They output a distribution on group fair matchings so that a randomly picked matching satisfies the individual fairness constraints.

We note that the term {\em proportional fairness} has also been used in literature in a way different from ours. For instance, in \cite{suhr2019_ride_hailing}, online matchings for ride-hailing platforms have been considered, and the term proportional fairness is used to indicate that the utility that each agent gets should be proportional to the time that the agent spends on the ride-hailing platform. Also, the term {\em diversity} has been used in various ways in literature (see e.g. \cite{CDKV16,FGJPS20}). In \cite{BCHSZ18}, both sides of the bipartition belong to various groups, referred to as {\em types} and {\em blocks}. They give a max-utility assignment that satisfies all the fairness constraints, given only in terms of upper bounds. 

The online model for the bipartite matching problem was studied in \cite{kvv_online} and they gave a randomized algorithm with a competitive ratio $1-1/e$. This was later generalized to online bipartite matchings with concave returns in \cite{devanur2012online}. In \cite{DSSX19}, fairness and diversity in online bipartite matchings has been achieved via submodular weight function.

Apart from matchings, fairness constraints have been considered in clustering problems \cite{EBSD21}, where fair clustering has been reduced to assignment problems and bicriteria approximations are shown \cite{BKKRS019,bera_fair_clustering}.


\section{The \propfair\ Problem}\label{sec:prop-fair}
We discuss the Proportionally Fair Matching problem in this section.
Here, we want to maximize
the number of items matched to those platforms whose proportional fairness constraints are met.

\subsection{Algorithm with no violation of constraints}
Given an instance $G$ of the \PFM\ problem, the algorithm involves constructing an instance of the hypergraph matching problem on a hypergraph $H$ with vertex set $V$ and edge set $F$ as follows:
$V=\{v_t\mid a_t\in A\}\cup \{p_j\mid p_j\in P\}$	
 Thus, corresponding to an item $a_t \in A$, we add a vertex $v_t$ to the set $V$. For every platform $p_j$ 
 we add a vertex $p_j$ to $V$. We denote by  $p_j$ for a vertex in the hypergraph corresponding to the platform $p_j$.
 Furthermore, we let $C_j^{(i)}=C_i\cap N(p_j)$.
The edge set is $F=\{	\{p_j\} \cup S\mid S\subseteq N(p_j),  \ell_j\leq \abs{S}\leq u_j \quad\forall i\in[m] \ \  \alpha_j^{(i)}  \abs{S} \leq \abs{S\cap C_j^{(i)}}  \leq \beta_j^{(i)} \abs{S}\}$

Thus, there is a hyperedge in $H$ corresponding to each possible assignment of items to $p_j$ that meets the proportional fairness constraints of $p_j$ and all its groups.
Further, by the above construction, observe that any assignment that satisfies the platform must
	correspond to a hyperedge. The size of the largest hyperedge in $H$ is $u$ where $u = \max_j u_j$. 

By Proposition~\ref{prop:folklore-greedy}, this naive reduction to hypergraph matching described above leads to
an approximation ratio of $O(u)$ without violating any of the proportional fairness constraints.
Given the hardness of approximation for hypergraph matching we do not expect to improve the approximation
guarantee using this reduction.

\subsection{Improved approximation with violation of constraints: }
We would like an approximation factor that depends on
$\ell_j$ instead of $u_j$ since the former is typically much smaller
than the latter e.g. in student course allocation,
the number of students required to offer a course is typically small
whereas the maximum capacity of the course may be much larger.
In order to get an improved approximation factor, we pay a price in terms of a slight violation of the fairness constraints.
 In the proof of Theorem~\ref{thm:prop-fair} below, we construct a hypergraph $H$ such that the hyperedges corresponding to a platform $p_j$ have size $\ell_j$, thereby resulting in an approximation factor $O(\ell)$ 
 where $\ell = \max_j{\ell_j}$.

\begin{proof}[Proof of Theorem~\ref{thm:prop-fair}]

We reduce the given instance $G$ to a hypergraph matching problem on a suitably constructed hypergraph $H$. We then use the greedy algorithm to compute a maximal matching $M_H$ in $H$.
However, owing to rounding errors, the corresponding matching $M$ in $G$ may not exactly satisfy the constraints.

{\bf Construction of the hypergraph $H$:} 
	Recall that in this problem, we only consider the case where each item belongs to exactly one group
	per platform.
	For every item $a_i$, we create a vertex $v_i$ in $H$.
	For every platform $p_j$, we create $t$ vertices $u_j^{(1)},\ldots,u_j^{(t)}$ where $t=\left \lfloor \frac{u_{j}}{\ell_j} \right \rfloor $.
	We add the following hyperedges for platform $p_j$.
	\begin{align*}
		\bigg \{\{u_j^{(k)}\} \cup S \mid 1\leq k\leq t, S\subseteq N(p_j), \abs{S}=\ell_j, 
		S \text{ satisfies Eqn~\ref{eqn:fairness-condition-prop}} \bigg\}.
	\end{align*}
	
		\begin{align}\label{eqn:fairness-condition-prop}
		\alpha_j^{(i)} \ell_j-3  \leq & \abs{S\cap C_j^{(i)}} \leq  \beta_j^{(i)} \ell_j+3.
	\end{align}
	We need to find a set $S$ satisfying the above property. 
 For this, we keep adding items from a group $C_j^{(i)}$ to $S$ until
	$|S\cap C_j^{(i)}|\geq \alpha_j^{(i)}\ell_j-3$.
	We repeat this for all the groups.
	To ensure $|S|=\ell_j$, we add items arbitrarily
	so that $|S\cap C_j^{(i)}|\leq \beta_j^{(i)}\ell_j+3$ for each group $C_j^{(i)}$. 
    It is easy to see that if such a set $S$ exists, the above process must find it, and the disjointness of groups is crucial here.
    
 \noindent{\bf Improving the running time to $O(n^2)$:}
 The construction of $H$ needs time $n^{\ell}$. 
	However, we improve it to $O(n^2)$ as follows. Instead of explicitly constructing all the hyperedges of $H$, we create
	a collection of disjoint hyperedges by constructing hyperedges one by one, ensuring that the hyperedge being created is disjoint from the previously created ones. This results in a simple greedy matching $M_H$ of $H$. Algorithm~\ref{alg:prop-fair} shows this directly for the instance $G$.

	\begin{algorithm}
	\caption{Improved algorithm for \propfair}\label{alg:prop-fair}
	\begin{algorithmic}[1]
	\For{every platform $p_j$ that arrives online}
	\While{$\exists~ S\subseteq N(p_j)$ of size $\ell_j$ satisfying Equation~\ref{eqn:fairness-condition-prop}}
	    \State Match all items from $S$ to $p_j$ and remove them.
	\EndWhile
	\EndFor
	\end{algorithmic}
\end{algorithm}

	Let the matching $M_H$ contain $t'\leq t$ hyperedges corresponding to platform $p_j$, each containing a distinct vertex from $u_j^{(1)}, \ldots u_j^{(t)}$. Note that the items in these hyperedges are assigned to $p_j$ in the corresponding matching $M$ in $G$. Let this set be $M_j$. By the construction of $H$, we observe that if $t'\ge 1$, we have $\ell_j \le |M_j| \le u_j$. Moreover since every hyperedge satisfies Eq.~\ref{eqn:fairness-condition-prop}, we 
	can add Equation~\ref{eqn:fairness-condition-prop} across all the $t'$ hyperedges to get
	\begin{align*}
		\alpha_j^{(i)} \abs{M_j}-3t'  \leq &\abs{M_j\cap C_j^{(i)}}  \leq  \beta_j^{(i)} \abs{M_j} +3t'\\
		\left(\alpha_j^{(i)} -\frac{3}{\ell_j}\right)\abs{M_j}  \leq &\abs{M_j\cap C_j^{(i)}}  \leq  \left(\beta_j^{(i)} +\frac{3}{\ell_j}\right) \abs{M_j}.
	\end{align*}
	Thus, we violate the constraints by at most a factor of $\frac{3}{\ell_j}$.
	



	
We show the desired approximation ratio of $2(\ell+1)$ as follows. In Lemma~\ref{lem:prop-fair-correctness} below, we show that for any optimum matching $M_{OPT}$ in $G$, there exists a matching $\sol_H$ in $G$ such that $|\sol_H|\geq \frac{|M_{OPT}|}{2}$. Then we give a transformation to get a matching $\sol_H$ in $H$. The $2(\ell+1)$ approximation then follows from Proposition~\ref{prop:folklore-greedy}. The final solution
satisfies the relaxed fairness constraints in Equation~\ref{eqn:fairness-condition-prop}. 
\end{proof}
	
\begin{lemma}\label{lem:prop-fair-correctness}
	The hypergraph $H$ contains a matching $\sol_H$ such that
	the number of items covered by hyperedges in $\sol_H$ is at least half
	the value of the optimum of original instance of \propfair
	and violates the fairness constraints by at most
	a factor of $\frac{3}{\ell_j}$
	for each platform $p_j$. 
\end{lemma}

\begin{proof}
	Let $M_{OPT}$ be the optimum matching in $G$.
	We focus our attention on some platform $p_j$.
    Let it be matched
	to a set $M_{OPT,j}$ of items in $M_{OPT}$.
    Let $\abs{M_{OPT,j}}=q\ell_j+r$ for non-negative integers $q,r$ where $r< \ell_j$ and $q>0$.
    We construct a solution $\sol$ whose restriction to the neighbours of $p_j$, $\sol_j$, satisfies $\abs{\sol_j}=q\ell_j$ matched items.
    Further, for each group $C_j^{(i)}$ of $p_j$, $\sol_j^{(i)}$ satisfies the following slightly relaxed constraint.
	For each $i$, let $M_{OPT,j}^{(i)}=M_{OPT,j}\cap C_j^{(i)}$.
	Then $\sol_j$ needs to satisfy 
	$\abs{\sol_j}=q\ell_j$ and
	\begin{align*}
	    &\abs{\sol_j^{(i)}} = \left \lceil \frac{q\ell_j}{q\ell_j+r}\abs{M_{OPT,j}^{(i)}} \right \rceil \text{ or }
	    &\abs{\sol_j^{(i)}} = \left \lfloor \frac{q\ell_j}{q\ell_j+r}\abs{M_{OPT,j}^{(i)}} \right \rfloor.
	\end{align*}
	This can easily be done by removing vertices one by one from $M_{OPT,j}^{(i)}$ for
	various $i$.
	Then we have
	\begin{align*}
		\frac{q\ell_j}{q\ell_j+r}\abs{M_{OPT,j}^{(i)}} -1  \leq & \abs{\sol_j^{\textcolor{black}{(i)}}}  \leq \frac{q\ell_j}{q\ell_j+r}\abs{M_{OPT,j}^{(i)}} +1
	\end{align*}
	\begin{equation}\label{eqn:SH_hypergraph}
	    \implies \alpha_j^{(i)} q\ell_j - 1 \leq   \abs{\sol_j^{(i)}} \leq  \beta_j^{(i)} q\ell_j+1.
	\end{equation}	
 Thus, there is a solution of size $q\ell_j$ that violates the lower and upper bounds
	by at most $1$.
	Note that we can repeat this for each satisfied platform, and this process may possibly reduce the size of the matching by a factor of $2$, i.e., $\abs{\sol}\geq \frac{1}{2}\abs{M_{OPT}}$.
    This is because, the loss is $r \leq \frac{q\ell_j+r}{2} $ because $r<\ell_j$ and
    $q>0$.
	
	Now, we want to construct a matching $\sol_H$ in $H$ from the solution $\sol$. 
		To do this, we divide the above $q\ell_j$ items in $\sol_j$ into $q$ hyperedges  
	of size $\ell_j$ each such
	that each hyperedge satisfies Equation~\ref{eqn:fairness-condition-prop}. These hyperedges together form the matching $\sol_H$ in $H$. Through the arguments given below, we show that this is
	possible.
	
	We do this by setting up a bipartite matching instance between items and hyperedges.
	We create a vertex for each item and hyperedge.
	Consider the $x$ items from some group $C_j^{(i)}$ that we need to assign to $q$ hyperedges.
	We fractionally match ${x_i}/{q}$ items from that group to each hyperedge.
    Observe that this can be done in a way so that each hyperedge has at most two fractional items 
    matched to it from one group: We imagine this as dividing the real interval between $[0,x]$
    into $q$ divisions. Each division can then be viewed as the union of a set of `integral' intervals of the form $[z,z+1]$ for some integer $z$ along with
    at most two `fractional' intervals (corresponding
    to items being fractionally being matched).
    We then repeat this for all the groups.
    For example, to assign $7$ items $\{i_1,\ldots, i_7\}$ to $3$ hyperedges $f_1,f_2,f_3$,
	we assign $i_1,i_2$ and $\frac{1}{3}$ of $i_3$ to $f_1$,
	the remaining $\frac{2}{3}$ 
	of $i_3$, whole $i_4$, and $\frac{2}{3}$ of $i_5$ to $f_2$, and the remaining $\frac{1}{3}$ of $i_5$ and whole of $i_6,i_7$ to $f_3$
	
	Now, we consider the fractional bipartite matching between the items and hyperedges.
	The edges have weight corresponding to the fraction of the item that was
	assigned. 
	Observe that this graph admits a fractional matching 
	with the property that on the item side, 
	the weights of edges adjacent to a vertex
	sum up to $1$ and on the side of hyperedges, the weights of the adjacent
	edges of every vertex sum up to $\ell_j$.
	From the integrality of the bipartite matching polytope
	\cite{plummer1986matching},
	there is an underlying integral matching such that
	every item is assigned to exactly $1$ hyperedge and every hyperedge is assigned
	exactly $\ell_j$ items.
	We create our hyperedges based on this matching.
	
	Since there are only two fractional items from a group 
	with an edge to a hyperedge, it can only
	violate the constraints by at most an additive factor of $2$. 
    Thus, we end up with hyperedges $S$ that satisfy the constraint
	\begin{align*}
		\frac{\abs{\sol_j^{(i)}}}{q}-2  \leq & \abs{S\cap C_j^{(i)}} \leq  \frac{\abs{\sol_j^{(i)}}}{q}+2
		\implies \alpha_j^{(i)} \ell_j-\frac{1}{q}-2  \leq \abs{S\cap C_j^{(i)}} \leq  \beta_j^{(i)} \ell_j+\frac{1}{q}+2.
	\end{align*}
	The hyperedges here, with addition of a distinct $u_j^{(i)}$ vertex, $1\leq i\leq q$,
	form a disjoint collection of hyperedges in $H$, thus giving the matching $\sol_H$ in $H$.
	Since our algorithm will not know the value of $q$ beforehand, we choose the worst case
	of $q=1$.
	Thus, the hyperedges satisfy
	\begin{align*}
		\alpha_j^{(i)} \ell_j-3  \leq & \abs{Z\cap C_j^{(i)}} \leq  \beta_j^{(i)} \ell_j+3.
	\end{align*}
 
\end{proof}


We note that the high-level idea of breaking a problem with group-fairness constraints into ``smaller''problems has been studied in the context of other problems as well such as fair clustering \cite{CKLV17}, fair rankings~\cite{GDL21} etc. However, doing so is problem specific and there is no known generic way of doing this for any problem. 

\section{Matchings with Diversity Constraints}\label{sec:lb}
We give a reduction from the \matchlb\ problem to the {\em Hypergraph Matching Problem} which implies an approximation
algorithm for the \matchlb\ problem, and thereby prove Theorem~\ref{thm:intro-matching-lb}. As in Section~\ref{sec:prop-fair}, we use the notation $C_j^{(i)}=N(p_j)\cap C_i$.

\begin{algorithm}
\begin{algorithmic}[1]
	\caption{Algorithm for \matchlb}\label{alg:max-satisfied}
	\For{platform $p_j$ that arrives online}
	
	\State Greedily construct a set $S\subseteq N(p_j)$ such that
		\begin{align*}
		\ell_j\leq \abs{S}\leq \max( \ell_j, \sum_i \ell_j^{(i)}) \textrm{ and }
		\abs{S\cap C_j^{(i)}}\geq  \ell_j^{(i)} \quad \forall~i\in [m]
		\end{align*}
	\State Match all items from $S$ to $p_j$ and remove them.
	
	\EndFor
\end{algorithmic}

\end{algorithm}

\begin{proof}[Proof of Theorem~\ref{thm:intro-matching-lb}]
Given an instance of \matchlb on a bipartite graph $G$ with parts $(A,P)$, we construct an instance of the hypergraph matching problem on a hypergraph $H$ with vertex set $V$.
	We add one vertex $v_i$ in $V$ corresponding to every item $a_i\in A$, and a new vertex $u_j$ in $V$ for each platform $p_j\in P$. Let $N(p_j)$ be the neighbourhood of platform $p_j$ in $G$.
	Then we add the following hyperedges to the graph

 	\begin{align*}
	\bigg\{	\{u_j\} \cup S\mid S\subseteq N(p_j),& \ell_j \leq \abs{S} \leq \max(\ell_j,\sum_k\ell_j^{(k)}) \quad \forall~k, \abs{S\cap C_j^{(k)}}\geq \ell_j^{(k)} \bigg\}.
	\end{align*}

Thus, there is a hyperedge in $H$ corresponding to each possible assignment of items to $p_j$ that satisfies the lower bounds of $p_j$ and all its groups.
	Observe that the largest hyperedge in the resulting hypergraph has size $\ell+1$. 
	It is immediate that every hyperedge in the matching corresponds to an assignment of 
	items to a platform that satisfies the platform.
	Further, observe that any assignment that satisfies the platform must
	correspond to a hyperedge.

	As stated earlier, we can now use the hypergraph matching greedy algorithm.
 Observe
	that the number of hyperedges can be as large as $O(n^\ell)$ and hence the time complexity would be as high too.
	However, we achieve $O(n^2)$ time complexity as follows. 
	Instead of explicitly constructing all the hyperedges, we keep constructing and picking one arbitrary hyperedge at a time that is disjoint from the previous ones.
	This is given in Algorithm~\ref{alg:max-satisfied} in terms of the setting in \matchlb. Thus for each platform $p_j$, we keep adding items from each group to a set $S_j$ until 
	the lower bound of the group is met by the items in $S_j$.
	Once the lower bounds of all the groups are met, if $|S_j|<\ell_j$,
	we add items arbitrarily to $S_j$ until $|S_j|=\ell_j$. 
 We ensure that an item is picked at most once.
	In the worst case, $|S_j|=\max(\ell_j,\sum_k\ell_j^{(k)})$.
	The theorem then follows via Proposition~\ref{prop:folklore-greedy}. 
\end{proof}



\section{Hardness Results}\label{sec:hardness}
We prove Theorem~\ref{thm:hardness-kdim} by giving a reduction from the {\em hypergraph matching} problem to \matchlb. 
The hardness of approximation for hypergraph matching, shown in \cite{hazan03_kdm}, then implies that the 
\matchlb is \sfNP-hard to approximate within a factor of $\frac{\ell}{\ln \ell}$. This hardness result holds even when there are no groups, and all the platforms have the same lower bound $\ell$. 

\begin{proof}[Proof of Theorem~\ref{thm:hardness-kdim}]
The hardness reductions are described below.

{\bf Reduction for the \matchlb\ problem: }
The reduction involves constructing an instance of the \matchlb\ problem from a given hypergraph matching instance $H=(F,V)$, which is a $k$-uniform hypergraph.
We construct a bipartite graph $G$ with parts $(A,P)$ as follows: Set $\ell=k$. The set $A$ contains an item $a_i$ for every vertex $v_i\in V$, the set $P$ contains a platform $p_j$ for every $e_j\in F$, and the edges of $G$ are all pairs $(a_i, p_j)$ such that $v_i\in e_j$ 
	The lower bound of $p_j$ is set to the size of $e_j$ i.e. $\ell$, and we have no groups in this case.

Now we show that, for every matching $M$ in $H$ of size $r$, there is an assignment of items to platforms that satisfies the lower bounds of $r$ platforms and vice versa.
Suppose there is a matching $M$ in $H$ of size $r$. Then for every hyperedge $e_j\in M$, we assign
items $\{a_i \mid v_i\in e_j\}$ to the corresponding platform $p_j$. This meets the lower bound of $p_j$, since the size of $e_j$ is equal to the lower bound of $p_j$ for each $j$.
Thus, since every $e_j$ has exactly $k$
vertices, and the edges in $M$ do not share vertices, we get an assignment
	of items to platforms that satisfies the lower bounds of $r$ platforms.
Similarly, suppose there is an assignment of items to platforms that satisfies
the lower bounds of $r$ platforms. Hence the assignment results in $r$ platforms that are satisfied. Let $p_j$ be some platform that is satisfied.
Then we choose the corresponding hyperedge $e_j$ into the solution. Since each item is matched to exactly one platform, corresponding hyperedges in the solution are disjoint.
This gives us a matching in the hypergraph, thereby establishing a one-to-one correspondence between matchings in $H$ and the matchings in $G$ satisfying all the lower bounds.

\noindent{\bf Reduction for the \propfair\ problem: }
The above reduction also gives a hardness for the \propfair\ problem as follows. After constructing the bipartite graph $G$ as above, for each platform $p_j$, set $u_j=\ell_j=k$. Further, set the number of groups to be $1$. Thus we have $C_j^{(1)}=N(p_j)$ for each platform $p_j$.
Also, set $\alpha^{(1)}_j=0,\beta^{(1)}_j=1$. It can be easily seen that the number of
items matched to satisfied platforms is precisely $k$ times the number
of satisfied platforms. The one-to-one correspondence between the set of matchings in $H$ and the set of matchings in $G$ that satisfy the proportional fairness constraints can be established as in the case of the \matchlb\ problem.
\end{proof}

\section{Analysis on Random graphs}\label{sec:exp-random-graphs}
\begin{proof}[Proof of Theorem~4]

We consider the Erd\H{o}s-R\'enyi random graph on the vertex set $A\cup P$. Moreover, we consider the case where each item belongs to exactly one group. Suppose there are $\Delta$ groups, and there are $n$ items per group. Hence there are $n\Delta$
items overall i.e. $|A|=n\Delta$. Let each platform have the lower bound $\ell$ for each group. Thus, the number of satisfied platforms can be at most $\frac{n}{\ell}$, and hence we take $|P|=\frac{n}{\ell}$.
Let $\rho$ be the probability
that any given edge is in the
graph. We will fix the value of $\rho$ later.

Recall that, in the execution of Algorithm~\ref{alg:max-satisfied} from Section~\ref{sec:lb}, 
the platforms are considered in some fixed order independent of the instance.
For each platform $p_j$, among the items from $N(p_j)$ which are not yet matched to any platform, the algorithm arbitrarily matches $\Delta\ell$ items to $p_j$, with exactly $\ell$
from each group. If this is not possible, it leaves platform $p_j$
unmatched.

Consider the stage in the execution of Algorithm~\ref{alg:max-satisfied} when it has considered $i-1$ platforms so far. 
Let $m_i$ be the number of platforms whose lower bounds
have been satisfied so far. 
Then we want to find the probability that the $i^{th}$ platform $p_i$ can be satisfied by the items that are yet unassigned. 
Note that there are $n-m_i\ell$ unmatched items from each group at this stage.
Platform $p_i$ can be satisfied if at least $\ell$ unmatched items from every group are in $N(p_i)$.
Let $X_j$ denote this event for the $j$th group, $1\leq j\leq \Delta$ where groups are ordered $1,\ldots,\Delta$ arbitrarily.
For a fixed $j$, 
\begin{align*}
  Pr[X_j] \geq 1-&(1-\rho)^{n-m_i\ell}-\ldots \binom{n-m_i\ell}{\ell-1}(1-\rho)^{n-(m_i+1)\ell+1}\rho^{\ell-1}\\
    &\geq 1-\ell (1-\rho)^{n-m_i\ell} \geq 1-\ell e^{\rho(m_i\ell-n)}
\end{align*}
Across all the $\Delta$ groups, the probability that at least $\ell$ unmatched items from each group belong to $N(p_i)$ is given by
\begin{align*}
   Pr[\bigcap_{j=1}^\Delta X_j] \geq \left(1-\ell e^{\rho(m_i\ell-n)}\right)^\Delta \geq 1-\Delta \ell e^{\rho(m_i\ell-n)}.
\end{align*}
Note that $m_i\leq i$. We choose $\rho=\frac{\alpha\log n}{n}$ for some constant $\alpha$.
Then, whenever $i\leq (1-\epsilon)\frac{n}{\ell}$, 
the above probability
is $1-\frac{1}{n^{\alpha\epsilon}}$ which is 
$1-\left(\frac{1}{n^d}\right)$ 
for constant $d=\alpha\epsilon$. 
Thus, with $1-\frac{1}{n^d}$ probability, the $i$th platform can be satisfied.
Taking a union bound over the first $(1-\epsilon)\frac{n}{\ell}$ platforms,
it follows that the algorithm can satisfy all of them with at least $1-\frac{1}{n^{d-2}}$ probability.
Thus, the greedy algorithm achieves an approximation ratio of at least $(1-\epsilon)$
with probability $1-\frac{1}{n^c}$ where $c=d-2$.
\end{proof}

\section{Experiments}
\label{sec:experiments}

We present the results of our experiments for Algorithm~\ref{alg:prop-fair} and  Algorithm~\ref{alg:max-satisfied}  for the \propfair\ and the \matchlb\ problems  respectively. 
We evaluate the algorithms on two kinds of datasets. The first one is a smaller
dataset containing anonymized data of the course allocation
process at IIT Madras, labelled Real-1 through Real-3.
The second is synthetic data generated using a random process 
resembling a random graph generated from an Erd\H{o}s-R\'enyi
model. Our code and dataset are available at \cite{repository-link}.

\noindent {\bf Data Sets and Setup:} Each of the three real-world datasets (Real-1, Real-2, Real-3), have around 3000 students
and 100 courses. Each course has a lower quota of $5$, 
denoting the minimum number of students needed to operate the course. As an instance of \matchlb, we would like to maximize the number of courses that satisfy this requirement.
The courses are all from an elective category, so each student is assigned only
one course.
Students are partitioned into groups based on their majors, and
there are $5$ groups overall.
The synthetic datasets contain 250 courses and 10,000 students.
These were generated as follows. For every student, a degree was chosen uniformly at random between 1 and an input parameter. All experiments were done on a desktop running 64-bit Windows using a 3.6 GHz Intel i7-7700 processor with 32.0 GB RAM. 

In our experiments for both problems our algorithm is denoted as ALG1 in the respective sections. We use two standard heuristics that
improve the performance in practice, though they do not improve the worst-case theoretical guarantees because of the hardness results.
The first heuristic (denoted as ALG2) is to prioritize matching the lowest-degree item to a platform when considering its neighbors.
The second one (denoted as ALG3) is to use augmenting paths 
when a platform is left with no unmatched
items among its neighbors.
We compare the optimum value obtained via solving an ILP with the
output of our algorithm. We present our ILPs for the two problems in the sections below.

\subsection{\propfair}
We use the following ILP for the proportionally fair matching problem:
The input bipartite graph is $G=(A\cup P,E)$
\begin{eqnarray}
\textrm{Maximize }\sum_{(a, p)\in E} x_{(a,p)} & &\textrm { subject to }\nonumber\\
 \alpha_p^{(k)} \cdot \sum_{(a,p)\in E} x_{(a,p)} & \le &\sum_{(a,p)\in E, a\in C^{(k)}_p} x_{(a,p)}\label{eqn:lb} \\ 
\sum_{(a,p)\in E, a\in C^{(k)}_p} x_{(a,p)} & \le & \beta_p^{(k)} \cdot \sum_{(a,p)\in E} x_{(a,p)} \label{eqn:ub} \\
\sum_{(a,p)\in E} x_{(a,p)} & \le & 1 \textrm{ for each }a\in A \label{eqn:assignment} \\
x_{(a,p)}  \in  \{0,1\} & &\textrm{ for each }(a,p)\in E \nonumber
\end{eqnarray}

The proportional fairness constraints (\ref{eqn:lb}) and (\ref{eqn:ub}) are for every group $C_p^{(k)}$ for each platform $p$ whereas the constraint (\ref{eqn:assignment}) encodes that each applicant is matched to at most one platform.\footnote{This ILP can be found in the file named ILP\_student\_max.py in our submitted source code.}
We run our experiments for \propfair\ on a smaller dataset owing
to a limitation in computational resources. We have 100 courses and
2000 students. Every course has the same set of groups.
There are 20 groups, each containing 100 students.
Every course has a lower bound of 10 students overall.
Every course $p_j$ and every group $k$ have $\alpha_j^{(k)}=0.025$
and $\beta_j^{(k)}=0.1$.
We  note that our algorithms solve a relaxed instance.
This explains the strange phenomenon that our algorithm outperforms
the OPT in some instances. The OPT is the optimal
solution to a more constrained instance.

\begin{figure}[ht]
    \centering
    \scalebox{0.8}{\includegraphics[width=\linewidth]{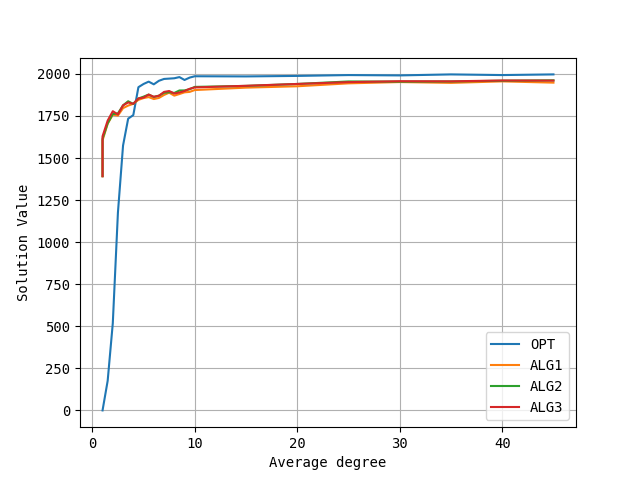}}
    \caption{\propfair: Solution value vs average degree of the random graph.}
\end{figure}

\begin{figure}[ht]
    \centering
    \scalebox{0.8}{\includegraphics[width=\linewidth]{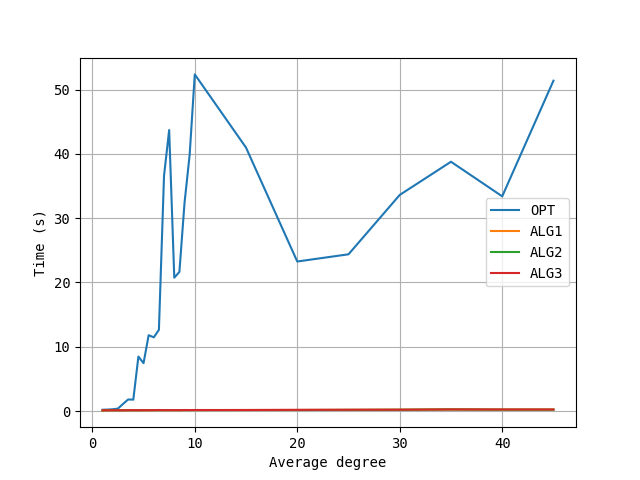}}
    \caption{\propfair: Running time vs average degree of the random graph.}
\end{figure}

\subsection{\matchlb}
Following is the ILP for the diverse matching problem: 

\begin{eqnarray*}
\textrm{Maximize }\sum_{p\in P} y_p & &  \textrm { subject to }\nonumber\\
\sum_{(a,p)\in E} x_{(a,p)} & \geq & \ell_p\cdot y_p \textrm{ for each }p\in P\\
\sum_{(a,p)\in E, a\in C^{(k)}_p} x_{(a,p)}& \geq & \ell^{(k)}_p \cdot y_p \textrm{ for each group }C^{(k)}_p \textrm{ of }p\\
\sum_{(a,p)\in E} x_{(a,p)} & \le & 1 \textrm{ for each }a\in A \\
y_p & \in & \{0,1\} \textrm{ for each }p\in P\\
x_{(a,p)} & \in & \{0,1\} \textrm{ for each }(a,p)\in E
\end{eqnarray*}

We run our experiments for \matchlb\ on  real world and synthetically generated instances where
the number of courses is 250 and the number of students is 10,000.
There are 20 groups,
each containing 500 students. Every course has a lower bound
of 2 for each group. We vary the average degree of the students
from 1 through 125.

All values were averaged over 15 runs. See Table~\ref{tab:table1} and Table~\ref{tab:table2} for solution value and running time comparison between our algorithms and OPT.

\begin{table}[!ht]
\centering
\begin{tabular}{@{}|ccccc|@{}}
\toprule
Dataset & OPT & ALG1   & ALG2   & ALG3   \\ \midrule
Real-1  & 34   & 29.53 & 31.73 & 33.20   \\
Real-2  & 31   & 27.73 & 29.40   & 30.86 \\
Real 3  & 31   & 26.07 & 28.73 & 30.00     \\ \bottomrule
\end{tabular}
\caption{\matchlb: Comparison of solution values of (ALG1, ALG2, ALG3) and an ILP that finds the Optimum (OPT). }
\label{tab:table1}
\end{table}

\begin{table}[!ht]
\centering
\begin{tabular}{@{}|crrrr|@{}}
\toprule
Dataset & \multicolumn{1}{c}{OPT1} & \multicolumn{1}{c}{ALG1} & \multicolumn{1}{c}{ALG2} & \multicolumn{1}{c|}{ALG3} \\ \midrule
Real-1  & 0.37                    & 0.11                    & 0.12                    & 0.12                      \\
Real-2  & 0.34                    & 0.11                    & 0.12                    & 0.12                     \\
Real 3  & 0.42                    & 0.13                   & 0.12                    & 0.12                     \\ \bottomrule
\end{tabular}
\caption{\matchlb: Runtime (in seconds) comparison between ALG1, ALG2, ALG3 and an ILP that finds the optimum (OPT). }
\label{tab:table2}
\end{table}

We observe that 
beyond a small threshold, 
the algorithm performed almost as well as the ILP. 
Our algorithm is much faster the ILP, particularly
for dense graphs. 
See Figure~\ref{fig:synth-prob1-b} for a trend on
performance of algorithm vs the average degree.
\begin{figure}[!ht]
    \centering
    \scalebox{0.8}{\includegraphics[width=\linewidth]{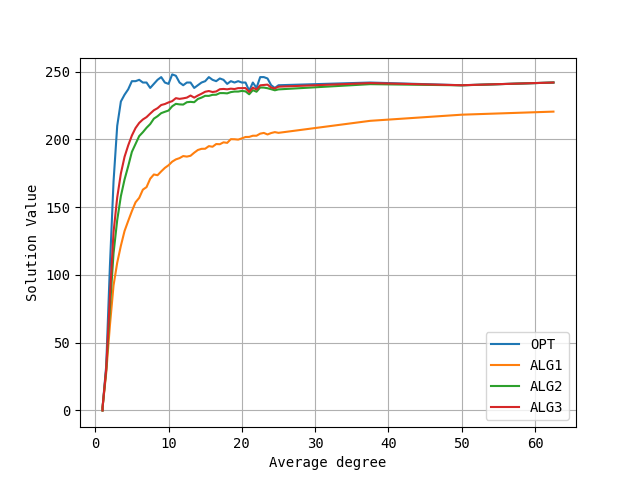}}
    \caption{\matchlb: Solution value vs average degree of the random graph for synthetic datasets.}
    \label{fig:synth-prob1-a}
\end{figure}
\begin{figure}[!ht]
    \centering
    \scalebox{0.8}{\includegraphics[width=\linewidth]{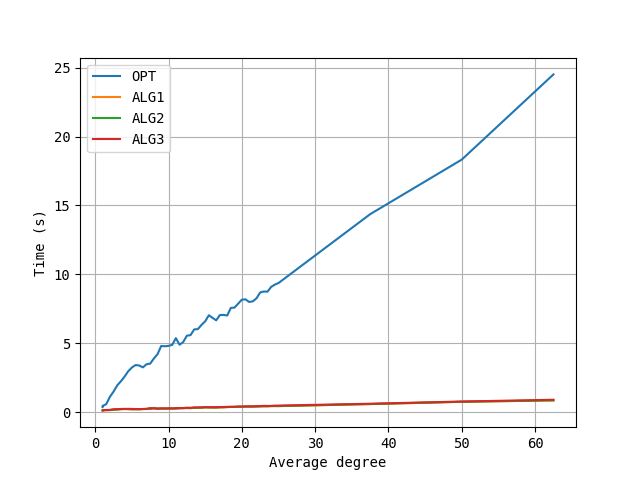}}
    \caption{\matchlb: Running time vs average degree of the random graph for synthetic datasets.}
    \label{fig:synth-prob1-b}
\end{figure}

\section{Discussion}
In this paper, we studied bipartite matching problems with proportional fairness constraints and diversity constraints. 
To the best of our knowledge, 
these constraints have not been considered particularly with the objective to maximize the number 
of platforms whose constraints are satisfied. 
Our algorithms exploited a connection to the hypergraph matching problem. 
Our algorithms generalize to the setting where each platform  defines its own groups on the set of items. 

Our approximation algorithm for the \propfair\ problem violates the fairness constraints by a small amount; obtaining a polynomial-time $O(\ell)$-approximation algorithm 
without violating the fairness constraints 
is an interesting open question.

\newpage
\bibliographystyle{plain}
\bibliography{references}
\end{document}